\documentclass[a4paper,11pt,twoside]{scrartcl}
\usepackage[utf8]{inputenc}
\usepackage[T1]{fontenc}
\usepackage[ngerman,english]{babel}
\usepackage[portrait,%
left=2.5cm,%
right=2.5cm,%
bottom=3cm]{geometry}
\usepackage{fancyhdr}
\usepackage{amsmath}
  \numberwithin{equation}{section} % number equations by section
\usepackage{amsfonts}
\usepackage{amssymb}
\usepackage{amsthm}
\usepackage{mathrsfs}
\usepackage{physics} % provides lots of physics macros
\usepackage{xparse} % needed for physics
\usepackage{url} % provides \url
\usepackage{footnote}
\usepackage{enumerate}
\usepackage[titletoc]{appendix}
\usepackage{bm}
\usepackage{microtype}
\usepackage[colorlinks=false, bookmarks, pdfborder={0 0 0}]{hyperref}
\usepackage{titling} % titlepage
  \pretitle{\begin{center}\huge\bfseries}
  \posttitle{\par\end{center}\vskip 0.5em}
  \preauthor{\begin{center}\large\scshape} % Author style
  \postauthor{\end{center}}
  \predate{\par\large\centering}
  \postdate{\par}

\theoremstyle{definition}
\newtheorem{definition}{Definition}[section]

\newtheorem{remark}[definition]{Remark}
\theoremstyle{plain}
\newtheorem{lemma}[definition]{Lemma}
\newtheorem{theorem}[definition]{Theorem}
\newtheorem{proposition}[definition]{Proposition}

\makeatletter
\newcommand{\eqnum}{\leavevmode\hfill\refstepcounter{equation}\textup{\tagform@{\theequation}}}
\makeatother

\newcommand{\inserttitle}{Liquid Drop Model for Nuclear Matter in the Dilute Limit}
\newcommand{\insertauthor}{Lukas Emmert, Rupert L. Frank, Tobias König}
\newcommand{\insertdate}{November 13, 2019}

\title{\inserttitle}
\author{\insertauthor}
\date{\insertdate}

\hypersetup{
  pdftitle={\inserttitle},
  pdfauthor={\insertauthor},
  pdfsubject={},
  pdfkeywords={},
}

\newcommand{\tl}{{\vartheta, L}}

\let\olddd\dd
\renewcommand*\dd[1]{\olddd{#1}}

\newcommand{\E}{\mathcal{E}}\newcommand{\h}{\mathcal{H}}
\newcommand{\F}{\mathcal{F}}
\newcommand{\N}{\mathbb{N}}
\newcommand{\Z}{\mathbb{Z}}
\newcommand{\R}{\mathbb{R}}

\newcommand{\e}{\mathrm{e}}

\newcommand{\Per}{\operatorname{Per}}
\newcommand{\till}{{\tilde{l}}}
\newcommand{\tilo}{{\tilde{\Omega}^*}}

\newcommand{\bl}{\textrm{\textbf{\textit{l}}}}

\begin{document}

\fancypagestyle{plain}{%
  \fancyhf{}
}
\thispagestyle{plain}

\maketitle

\begin{abstract}
  We consider the liquid drop model for nuclei interacting with a neutralizing
  homogeneous background of electrons. The regime we are interested in is when the
  fraction between the electronic and the nuclear charge density is small. We show that in this dilute limit the thermodynamic ground state energy is given to leading order by that of an isolated nucleus.
\end{abstract}

%%%%%%%%%%%%%%%%%%%%%%%%%%%%%%%%%%%%%%%%%%%%%%%%%%%%%%%%%%%%%%%%%%%%%%%%%%%%%%%%

\section{Introduction and Main Theorem}

%%%%%%%%%%%%%%%%%%%%%%%%%%%%%%%%%%%%%%%%%%%%%%%%%%%%%%%%%%%%%%%%%%%%%%%%%%%%%%%%

\renewcommand{\thefootnote}{${}$} \footnotetext{\copyright\, 2019 by
  the authors. This paper may be reproduced, in its entirety, for
  non-commercial purposes.}

Gamow's liquid drop model \cite{Ga} is a simple model in nuclear physics which has recently attracted a lot of attention in mathematics, see, for instance, \cite{ChPe1, ChPe2, BoCr, Ju2, KnMu, LuOt, FrLi, KnMuNo,  ChMu, Ju1, FrLi2019}.

We begin with the description of a single nucleus in this model. Possible shapes of a nucleus are (measurable) sets $\Omega\subset\R^3$ and their measure $|\Omega|$ is interpreted as the number of nucleons in suitable units. The energy of such a nucleus is, again in suitable units,
\begin{equation}
  \label{eq_definition free energy functional}
  \E[\Omega] := \Per(\Omega) + \frac{1}{2} \int_{\Omega} \int_{\Omega} \frac{\dd x \dd y }{|x-y|}.
\end{equation}
This leads to the variational problem of finding, for a given $A > 0$,
\begin{equation}
  \label{eq_definition free energy}
  E(A) := \inf \{ \E [\Omega] \,:\, \Omega \subset \R^3, |\Omega| = A \}.
\end{equation}
It is known \cite{FrLi} (see also \cite{KnMuNo}) that there is an $A^* > 0$ such that $\frac{E(A^*)}{A^*}
= \inf_{A > 0} \frac{E(A)}{A}$ and that there is a minimizing set $\Omega^*
\subset \R^3$ with $|\Omega^*| = A^*$ such that $\E[\Omega^*] = E(A^*)$. This set $\Omega^*$ is strongly conjectured, but not known, to be a ball. Physically, it corresponds to a nucleus with the greatest binding energy per nucleon, which is a certain isotope of nickel.

In this paper, we are interested in a system consisting of a large number of
nuclei interacting with a uniform background of electrons, as arises, for
instance, in the crust of a neutron star. As usual in statistical mechanics, we
confine the system to a box $[-L/2, L/2]^3$ and are interested in the
thermodynamic limit $L\to\infty$. For finite $L>0$, allowed nuclear configurations are described by measurable sets $\Omega \subset [-L/2, L/2]^3$ and their energy is
\begin{equation}
  \label{eq_definition functional E theta L}
  \E_\tl[\Omega] := \Per(\Omega) + \frac{1}{2} \int_{[-L/2, L/2]^3} \int_{[-L/2, L/2]^3} (1_\Omega(x) - \vartheta) \frac{1}{|x-y|} (1_\Omega(y)- \vartheta) \dd x \dd y.
\end{equation}
The parameter $\vartheta \in (0,1]$ here describes the quotient between the electron and the nucleon charge density. We are interested in the ground state energy
\begin{equation}
  \label{eq_definition number E theta L}
  E_\tl := \inf \Bqty{ \E_\tl[\Omega] \, : \, \Omega \subset [-L/2, L/2]^3, |\Omega| = \vartheta L^3 }.
\end{equation}
Note that the constraint $|\Omega| = \vartheta L^3$ means that we only consider neutral configurations.

Our main result concerns the behavior of the energy per unit volume,
$E_\tl/L^3$, in the dilute limit $\vartheta\to 0$. The crucial point is to establish this uniformly in $L$.

\begin{theorem}[Ground State Energy Asymptotics] \label{th_gse asymptotics}
There is a constant $C > 0$ such that the following bounds hold.
  \begin{enumerate}
    \item[(i)]
  For all $\vartheta \in (0, \frac12]$ and $L >0$ such that
      $\vartheta^{1/3}L \ge C$, we have
      \begin{align} \label{upper_bound intro}
        \frac{E_\tl}{\vartheta L^3} \le   \frac{E(A^*)}{A^*}
        + C \vartheta^{1/3} + \frac{C}{\vartheta^{1/3} L} .
      \end{align}
    \item[(ii)]
     For all $\vartheta \in (0,1]$ and $L >0$, we have
      \begin{align} \label{lower_bound intro}
        \frac{E_\tl}{\vartheta L^3} \,\geq\, \frac{E(A^*)}{A^*}- C \vartheta^{1/5}.
      \end{align}
  \end{enumerate}
\end{theorem}

% \begin{remark}
% In part (i), a similar result for $\vartheta \in [\frac12,1)$ can be obtained
% from our case $\vartheta \in (0,\frac12]$ using the transformation
% \begin{align*}
  % \vartheta \to 1 - \vartheta
  % \qquad \textnormal{and} \qquad
  % \Omega \to Q_L \setminus \Omega,
% \end{align*}
% which is a symmetry of the Coulomb term in the energy \eqref{eq_definition
% functional E theta L} and affects the perimeter term only to the subleading
% order $L^2$.
% \end{remark}

\begin{remark}
  Our result implies, in particular, that the thermodynamic limit
  \begin{align*} \label{thermodynamic_limit}
    e(\vartheta) := \lim_{L\to\infty} L^{-3} E_\tl,
  \end{align*}
  satisfies
  \begin{equation}
  \label{thermodynamic limit asymptotics}
  e(\vartheta) = \vartheta \left( \frac{E(A^*)}{A^*} + o(1) \right)
  \qquad\text{as}\ \vartheta\to 0 \,.
  \end{equation}
\end{remark}

\begin{remark}
  The bounds in Theorem \ref{th_gse asymptotics} give the asymptotics of
  the energy for $\vartheta$ close to $0$.
  By a simple symmetry argument our theorem yields analogous asymptotics
  for $\vartheta$ close to $1$.
  Namely, for $\vartheta \in [\frac12,1)$ and $(1-\vartheta)^{1/3}L \ge C$, we have
  \begin{align} \label{theta_large}
    - C (1-\vartheta)^{1/5} - \frac{6}{L}
     \le \frac{E_{\vartheta,L}}{(1-\vartheta) L^3} - \frac{E(A^*)}{A^*}
    &\le   C (1-\vartheta)^{1/3} + \frac{C}{(1-\vartheta)^{1/3} L}.
  \end{align}
  This follows from the fact that we have for $\Omega \subset Q_L$
  \begin{align*}
    \E_{\vartheta, L} [\Omega] = \E_{1-\vartheta,L}[Q_L\setminus\Omega]
    - \big(\mathcal H^2(\partial Q_L) - 2 \mathcal H^2 (\partial Q_L \cap \overline
    \Omega) \big),
  \end{align*}
  where the closure of $\Omega$ is taken in the measure theoretic sense.
  The term in parentheses is bounded in absolute value by $6L^2$.
  Therefore, \eqref{theta_large} follows from the bounds in Theorem \ref{th_gse
  asymptotics}.
\end{remark}

\begin{remark}
The power $1/5$ of $\vartheta$ in \eqref{lower_bound intro} is technical. It is an interesting question to decide whether the power $1/3$ in \eqref{upper_bound intro} is best possible. The assumption $\vartheta^{1/3} L \geq C$ and the corresponding remainder term in \eqref{upper_bound intro} are not severe restrictions in the thermodynamic limit and are imposed mainly for a simple statement.
\end{remark}

\begin{remark}
\label{remark reference to appendix}
  In the literature the Coulomb kernel $|x-y|^{-1}$
  in the energy \eqref{eq_definition free energy} is often replaced by the Green's
  function of the Laplacian on $Q_L$ with some boundary condition.
  Apart from the 'whole space' condition that we use, common choices are the
  periodic, the Dirichlet and the Neumann boundary condition.

In the Neumann case, the existence of the thermodynamic limit, together with the
sharp convergence rate of order $L^{-1}$, was shown in a remarkable paper by
Alberti, Choksi and Otto \cite{AlChOt}. In Appendix
\ref{se_thermodynamic_limit}, we use comparison arguments to extend their result
to all the boundary conditions mentioned above. In particular, this shows that the left side in \eqref{thermodynamic limit asymptotics} is independent of the choice of boundary values.
\end{remark}

Theorem \ref{th_gse asymptotics} significantly improves the main result of
Knüpfer, Muratov and Novaga \cite{KnMuNo}, who show a similar asymptotic
equality in the ultra-dilute limit $\vartheta\sim L^{-2}$, where the background
density vanishes in the limit $L\to\infty$. In contrast, we can perform first
the limit $L\to\infty$ and then $\vartheta\to 0$. In the regime $\vartheta\sim
L^{-2}$ screening does not yet play a role and controlling this phenomenon is,
in fact, one of the accomplishments in this paper. 

The even more dilute situation where $\vartheta\sim L^{-3}$ was considered by
Choksi and Peletier \cite{ChPe1}. Then, the leading order $E(A^*)/A^*$ in \eqref{upper_bound intro} and \eqref{lower_bound intro} should be replaced by $E(\vartheta L^3)/(\vartheta L^3)$. In this situation our upper bound \eqref{upper_bound intro} is not applicable (at least not if $\vartheta L^3$ is too small) and our lower bound \eqref{lower_bound intro} is not tight. However, a simple variation of our arguments would also cover this regime. On the other hand, \cite{ChPe1} also establishes a lower order correction.

It is conjectured that for small $\vartheta$ minimizers are given, at least in the bulk, by a periodic arrangement of nearly spherical sets. Our main result provides evidence for the latter prediction. If this conjecture is true, then, as the nucleon density $\vartheta$ tends to zero, the balls on the lattice should move infinitely far apart. Each one of the balls should therefore be asymptotically equal to an energy-per-volume minimizer of the full-space energy functional \eqref{eq_definition free energy functional}, and hence the energy per unit volume should be given to leading order by $\inf_{0 < |\Omega| < \infty} |\Omega|^{-1}\E[\Omega] = (A^*)^{-1} E(A^*)$. This intuition guides us in the proof of the upper bound \eqref{upper_bound intro}.

We emphasize that our result is valid independently of whether $\Omega^*$ is a ball or not. This is particularly relevant for the proof of the upper bound \eqref{upper_bound intro}. When $\Omega^*$ is a ball, or more generally, when the quadrupole moment of $\Omega^*$ vanishes, the proof of the upper bound \eqref{upper_bound intro} is straightforward following the above intuition. When the quadrupole moment of $\Omega^*$ does not vanish, we need to distort the lattice to achieve the required cancellation in the long range behavior of the Coulomb potential.

The problem of proving periodicity of minimizers in this and other,
multi-dimensional minimization problems is a well-known and long standing open
problem (crystallization conjecture). The strongest result about local order for
the present problem was shown in the work \cite{AlChOt} mentioned before.
Remarkably, for the present problem it was proposed in the physics literature
\cite{RaPeWi,HaSeYa} that there are phase transitions at
$0<\vartheta_{c1}<\vartheta_{c2}<1/2<\vartheta_{c3}=1-\vartheta_{c2}<\vartheta_{c4}=1-\vartheta_{c1}<1$,
where the dimensionality of the periodicity changes. For
$0<\vartheta<\vartheta_{c1}$, minimizers are expected to be sphere shaped and
arranged in a three dimensional lattice, for
$\vartheta_{c1}<\vartheta<\vartheta_{c2}$, minimizers are expected to be
cylinder shaped and arranged in a two-periodic lattice and for
$\vartheta_{c2}<\vartheta<\vartheta_{c3}$ minimizers are expected to be slab
shaped with respect to a one-dimensional lattice. For $\vartheta>1/2$ the
situation reverses (since $\vartheta\mapsto1-\vartheta$ corresponds to
$\Omega\mapsto\R^3\setminus\Omega$) and one expects a transition to cylindrical
holes and then to spherical holes. This phenomenon is sometimes referred to as
`nuclear pasta phases'. Numerically, one has $\vartheta_{c1}\approx 0.20$ and
$\vartheta_{c2}\approx 0.35$ \cite{OyKaHa84}. We refer to a recent result
\cite{GiSe} where the optimality of slab-like structures was rigorously
established in a multi-dimensional lattice model which, similarly to the present
model, contains an attractive short range term competing with a repulsive
long-range term, see also \cite{DaRu}.

The remainder of this paper consists of two sections and two appendices. The
first section deals with the upper bound \eqref{upper_bound intro} in case the
minimizer $\Omega^*$ is a ball and the second one with the lower bound
\eqref{lower_bound intro}.  In Appendix \ref{appendix_removing_symmetry} we describe
the necessary changes in the proof of the upper bound in the situation where the
minimizer $\Omega^*$ of $\E[\Omega]/|\Omega|$ as defined in \eqref{eq_definition
free energy functional} is not a ball.  Using \cite{AlChOt}, we furthermore show
in Appendix \ref{se_thermodynamic_limit} that the thermodynamic limit of the
ground state energy is independent of the choice of boundary conditions.

%%%%%%%%%%%%%%%%%%%%%%%%%%%%%%%%%%%%%%%%%%%%%%%%%%%%%%%%%%%%%%%%%%%%%%%%%%%%%%%%

\paragraph{Notation.} Since cubes with different sizes and centers will be a recurring tool in our analysis, it is convenient to introduce the following notation. For $r=(r_1, r_2, r_3) \in \R^3$, we set
\[ |r|_\infty = \max \{ |r_i| \, : \, i = 1,2,3 \}. \]
Then, for $r \in \R^3$ and $l > 0$, we define
\begin{equation}
  \label{eq_definition cube}
  Q_l(r) := \{ x \in \R^3 \, : \, |x - l r|_\infty < l/2 \},
\end{equation}
and $Q_l := Q_l(0)$.
Pay attention to the fact that by definition, $Q_l(r)$ is the cube of side length
$l$ centered at the point $l r$, not at $r$! In other words, to obtain $Q_l(r)$, one first takes a cube of unit side length centered at $r$ and then dilates it by the factor $l$. We found this slightly unusual definition better suited for our purposes.

%%%%%%%%%%%%%%%%%%%%%%%%%%%%%%%%%%%%%%%%%%%%%%%%%%%%%%%%%%%%%%%%%%%%%%%%%%%%%%%%

\section{Upper Bound of the Ground State Energy} \label{section upper bound}

%%%%%%%%%%%%%%%%%%%%%%%%%%%%%%%%%%%%%%%%%%%%%%%%%%%%%%%%%%%%%%%%%%%%%%%%%%%%%%%%

The purpose of this section is to prove the first statement of Theorem \ref{th_gse asymptotics}, which we restate here for convenience.

\begin{proposition}[Upper Bound] \label{pr_upper bound}
      There is a constant $C > 0$ such that, if
      $\vartheta^{1/3}L \ge C$ and $\vartheta \le \frac12$, we have
  \begin{align} \label{upper_bound}
    \frac{E_\tl}{\vartheta L^3} \le   \frac{E(A^*)}{A^*}
    + C \vartheta^{1/3} + \frac{C}{\vartheta^{1/3} L} .
  \end{align}
\end{proposition}

\begin{proof}
  For simplicity, we assume in the following the minimizer $\Omega^*$ of the whole space problem
  to be a ball $B(0, r_*)$ of the appropriate radius $r_*$ centered at 0. In the appendix it is explained how the proof has to be modified
  if this is not so.

  To prove Proposition \ref{pr_upper bound}, we construct, for every pair
  $(\tl)$, a suitable competitor set $\Omega_\tl$ for $E_\tl$. The idea is to take
  $\Omega_\tl$ to be given by a cubic lattice arrangement on $Q_L$ of sets
  $\Omega^*$. The period length $l >0$ of the lattice will be chosen so that the
  requirement $|\Omega_\tl| = \vartheta L^3$ is fulfilled.
  Since we want each box of side length
  $l$ to contain one copy of $\Omega^*$, for the mass density to be equal to
  $\vartheta$, we need to require $\vartheta l^3 = A^*$, or
  \[ l = {A^*}^{1/3} \vartheta^{-1/3}. \]

  Let ${\mathcal C_\tl} := \{ r \in \Z^3 \, : \, Q_l(r) \subset  Q_L \}$ be the set
  of lattice points $r$ such that the cubes $Q_l(r)$ are fully contained in $Q_L$.
  Let $N_\tl := \# {\mathcal C_\tl}$ denote the number of these cubes.

  We now define the set $\Omega_\tl$ to be the following disjoint union

  \begin{equation} \label{eq_definition Omega tl trial state}
    \Omega_\tl := \bigcup_{r \in {\mathcal C_\tl}} \pqty{ \lambda_\tl l r + \lambda_\tl \Omega^* },
  \end{equation}
  where the rescaling factor $\lambda_\tl$ is given by
  \begin{equation}
    \label{eq_definition rescaling factors}
    \lambda_\tl^3 = \frac{\vartheta L^3}{A^* N_\tl}.
  \end{equation}

  Note that the union in \eqref{eq_definition Omega tl trial state} is disjoint
  since $\vartheta \le \frac12$.
  (Indeed, we have $\frac12 l^3 \ge \vartheta l^3 = A^* = \frac{4\pi}{3} r_*^3$ and
  hence, $l > 2r_*$.)

  Informally, our construction of the competitor set $\Omega_\tl$ can thus be
  described as follows. We fill $Q_L$ with small boxes $Q_l(r)$ of side length
  $l$ as full as possible, place a copy of $\Omega^*$ in the middle of each box and enlarge the whole configuration slightly by the factor $\lambda_\tl$.

  The definition of $\lambda_\tl$ now ensures that the boxes $Q_{\lambda_\tl l}(r)$ cover $Q_L$ completely and that the mass constraint
  \begin{equation} \label{eq_mass constraint with lambda}
    |\Omega_\tl| = N_\tl A^* \lambda_\tl^3 = \vartheta L^3
  \end{equation}
  is fulfilled. Note also that with this choice, we even have \emph{local neutrality} of $\Omega_\tl$ on every box $Q_{\lambda_\tl l}(r)$, i.e. for every $r \in \mathcal C_\tl$,
  \begin{equation}
    \label{eq_local neutrality of Omega theta L}
    |\Omega_\tl \cap Q_{\lambda_\tl l}(r)| = \lambda_\tl ^3 A^*  = \vartheta \lambda_\tl^3 l^3 = \vartheta |Q_{\lambda_\tl l}(r)|.
  \end{equation}

  Since the number of boundary boxes is of order $\frac{L^2}{l^2}$, it is easy to see that $N_\tl$ satisfies the bounds
  \begin{equation} \label{eq_bounds on number of boxes}
    \frac{L^3}{l^3} \geq N_\tl \geq \frac{L^3}{l^3} - C \frac{L^2}{l^2},
  \end{equation}
  for some $C >0$ independent of $\vartheta$ and $L$.
  From \eqref{eq_definition rescaling factors}, we thus obtain the bound
  \begin{equation}
    \label{eq_bounds on lambda tl} 1 \leq \lambda_\tl^3 \leq \frac{\vartheta L^3
    }{A^*\qty(\frac{L^3}{l^3} - C \frac{L^2}{l^2})} = \frac{\vartheta}{A^*(l^{-3} - C L^{-1} l^{-2})} = \frac{1}{1 - C \frac{l}{L}} \leq 1 + C \frac{l}{L},
  \end{equation}
  and so, in particular, $\lim_{L \to \infty} \lambda_\tl = 1$.   \\
  In many situations below, to estimate subleading terms, the crude bound
  \begin{equation}
    \label{eq_bound on lambda tl crude}
    1 \le \lambda_\tl \leq 2,
  \end{equation}
  is enough. It follows from \eqref{eq_bounds on lambda tl} whenever $\frac{l}{L} = \frac{{A^*}^{1/3}}{\vartheta^{1/3} L} \leq C^{-1}$, for some universal $C>0$. \\

  Our proof of the bound \eqref{upper_bound} consists in computing in three separate steps the self-energy, the near-field and the far-field interaction energy of the set $\Omega_\tl$. That is, we split
  \begin{equation}
    \label{eq_splitting into self, near, far}
    \E_\tl [\Omega_\tl] = \E^\text{(self)}_\tl + \E^\text{(near)}_\tl + \E^\text{(far)}_\tl,
  \end{equation}
  by partitioning the double integral from the interaction term of $\E_\tl$. To simplify notation, we will write
  \begin{align}
    \till = \lambda_\tl l
  \end{align}
  in the rest of this proof. Hence, we define
  \begin{equation}
    \label{definition E self}
    \E^\text{(self)}_\tl := \Per(\Omega_\tl) + \sum_{r \in {\mathcal C_\tl}} \frac12 \int_{Q_\till(r)} \int_{Q_\till(r)} (1_{\Omega_\tl}(x) - \vartheta) \frac{1}{|x-y|}
    \pqty{1_{\Omega_\tl}(y) - \vartheta} \dd{x} \dd{y},
  \end{equation}
  and
  \begin{equation}
    \label{definition E near}
    \E^\text{(near)}_\tl := \sum_{(r,s) \in V_{\textnormal{near}}}
    \frac12 \int_{Q_\till(r)} \int_{Q_\till(s)} \pqty{1_{\Omega_\tl}(x) - \vartheta} \frac{1}{|x-y|}
    \pqty{1_{\Omega_\tl}(y) - \vartheta} \dd{x} \dd{y},
  \end{equation}
  where
  \begin{equation}
    \label{definition V near}
    V_{\textnormal{near}}:= V^{\textnormal{(near)}}_\tl := \big\{(r,s) \in
      {\mathcal C_\tl} \times {\mathcal C_\tl}\textnormal{ and } 1 \le |r-s|_\infty
    \le M \big\}.
  \end{equation}
  Here, $M \in \N$ is a number fixed throughout our proof (let us say  $M=10$).

  Lastly, we define
  \begin{equation}
    \label{definition E far}
    \E^\text{(far)}_\tl := \sum_{(r,s) \in V_{\textnormal{far}}} \frac{1}{2}
    \int_{Q_\till(r)} \int_{Q_\till(s)} \pqty{1_{\Omega_\tl}(x) - \vartheta} \frac{1}{|x-y|}
    \pqty{1_{\Omega_\tl}(y) - \vartheta} \dd{x} \dd{y},
  \end{equation}
  where
  \begin{equation}
    \label{definition V far}
    V_{\textnormal{far}} := V^{\textnormal{(far)}}_\tl := \big\{(r,s) \in {\mathcal C_\tl} \times {\mathcal C_\tl} \textnormal{ and } |r-s|_\infty > M \big\}.
  \end{equation}

  \textit{Step 1: Self-Energy.}
  Similarly to $\till$, we write
  \begin{align}
    \tilo = \lambda_\tl \Omega^*
  \end{align}
  here and in the rest of this proof. Since $\Omega_\tl$ consists of $N_\tl$ disjoint copies of $\tilo$, we have
  \begin{align*}
    \frac{\E^\text{(self)}_\tl }{\vartheta L^3}  & \leq \frac{N_\tl }{\vartheta
    L^3} \qty( \Per(\tilo) + \frac{1}{2} \int_{\tilo} \int_{\tilo} \frac{\dd x \dd y}{|x-y|}  +  \frac{\vartheta^2}{2} \int_{Q_\till} \int_{Q_\till} \frac{\dd x \dd y}{|x-y|} ).
  \end{align*}
  Since $\frac{N_\tl }{\vartheta L^3} = \frac{1}{A^* \lambda_\tl^3}$ by \eqref{eq_mass
  constraint with lambda}, we have
  \begin{align*}
    \frac{\E^\text{(self)}_\tl }{\vartheta L^3}  &\leq \frac{1}{A^*} \Big( \lambda_\tl^{-1} \Per(\Omega^*) + \frac{\lambda_\tl^2 }{2} \int_{\Omega^*} \int_{\Omega^*}  \frac{\dd x \dd y}{|x-y|} \dd{x} \dd{y} +  C \vartheta^2 \till^5 \Big)\\
    & \leq \frac{E(A^*)}{A^*} +  \frac{E(A^*)}{A^*} (\lambda_\tl^2 - 1) +  C \vartheta^2 l^5 \\
    & \leq \frac{E(A^*)}{A^*} + \frac{C}{\vartheta^{1/3} L} + C \vartheta^{1/3},
  \end{align*}
  where we used the bound $\lambda_\tl^2 - 1 \leq C(\lambda_\tl^3 - 1) \leq C
  \frac{l}{L} = C \frac{1}{\vartheta^{1/3} L}$ from \eqref{eq_bounds on lambda
  tl} for the last inequality. Moreover, recall $l \le \tilde l = \lambda_\tl l \le 2l$, from \eqref{eq_bound on lambda tl crude}.\\

  %%%%%%%%%%%%%%%%%%%%%%%%%%%%%%%%%%%%%%%%%%%%%%%%%%%%%%%%%%%%%%%%

  \textit{Step 2: Near Field Interaction.}
  Due to the periodicity of $\Omega_\tl$, we have
  \begin{align} \nonumber
    \frac{ \E_\tl^\text{(near)}}{\vartheta L^3} &=  \frac1{\vartheta L^3} \sum_{(r,s) \in V_{\textnormal{near}}}
    \frac12   \int_{Q_\till(r)} \int_{Q_\till(s)} \pqty{1_{\Omega_\tl}(x) - \vartheta} \frac{1}{|x-y|}
    \pqty{1_{\Omega_\tl}(y) - \vartheta} \dd{x} \dd{y}
    \\ \nonumber
    &= \frac1{\vartheta L^3} \sum_{(r,s) \in V_{\textnormal{near}}}
    \frac12 \int_{Q_\till} \int_{Q_\till} \pqty{1_{\tilo}(x) - \vartheta} \frac{1}{|r\till + x- s\till- y|}
    \pqty{1_{\tilo}(y) - \vartheta} \dd{x} \dd{y} \\
    \label{eq_upper_bound_near_field1}
    &\le \frac1{\vartheta L^3} \sum_{(r,s) \in V_{\textnormal{near}}}
    \frac12 \pqty{  \int_{\tilo} \int_{\tilo} \frac{1}{|r\till + x- s\till - y|} \dd{x} \dd{y}
    + \vartheta^2 \int_{Q_\till} \dd{x} \int_{Q_\till} \frac{1}{|y|} \dd{y} },
  \end{align}
  where we used the fact that the integral over the symmetric-decreasing function $1/|\cdot|$ is
  largest on the cube centered at $0$. This follows from the observation that three Steiner symmetrizations with respect to the coordinate directions $e_1$, $e_2$, $e_3$ transform any cube $Q_l(\mu)$ into the centered cube $Q_l(0)$.
  Furthermore, for every $r \neq s$ and $x,y \in \tilde \Omega^*$ we have
  \[ |r\till + x- s\till - y| \geq |r-s|\till - |x-y| \geq \lambda_\tl \pqty{l -
  \text{diam}(\Omega^*) } \geq \till /C. \]
  Here, we used $\frac12 l^3 \ge \vartheta l^3 = A^* = \frac{4\pi}{3} r_*^3$,
  which implies $l - 2r_* \ge l/C$.
  Hence, the right hand side of
  \eqref{eq_upper_bound_near_field1} is bounded from above by
  \begin{align} \label{eq_upper_bound_near_field2}
    \frac1{\vartheta L^3} &\sum_{r \in \mathcal C_\tl} M^3 \,
    \pqty{ \frac{C}\till + C\vartheta^2 \till^{\;\!5}}
    \le C \frac{L^3 \till^{-3}}{\vartheta L^3} \pqty{ \till^{-1} + \vartheta^2 \till^{\;\!5}}
    \leq C \vartheta^{1/3},
  \end{align}
  where we used the bound \eqref{eq_bounds on number of boxes}.
  For the last inequality, recall the choice $\vartheta l^3 = |\Omega^*| = A^*$
  and the bound $1 \le \lambda_\tl \leq 2$ from \eqref{eq_bound on lambda tl crude}. \\

  \textit{Step 3: Far Field Interaction.}
  Due to the periodicity of $\Omega_\tl$, we have
  \begin{align} \nonumber
    \E_\tl^{\text{(far)}} = &\sum_{(r,s) \in V_{\textnormal{far}}}
    \frac12 \int_{Q_\till(r)} \int_{Q_\till(s)} \pqty{1_{\Omega_\tl}(x) - \vartheta} \frac{1}{|x-y|}
    \pqty{1_{\Omega_\tl}(y) - \vartheta} \dd{x} \dd{y}
    \\ \label{eq_upper_bound_far_field1}
    =& \sum_{(r,s) \in V_{\textnormal{far}}}
    \frac12 \int_{Q_\till} \int_{Q_\till} \pqty{1_{\tilo}(x) - \vartheta} \frac{1}{|r\till + x- s\till- y|}
    \pqty{1_{\tilo}(y) - \vartheta} \dd{x} \dd{y}.
  \end{align}

  We now use the Taylor expansion
  \begin{align} \label{eq_taylor expansion multipole}
    \frac1{|a-b|} = \frac1{|a|} + \frac{a\cdot b}{|a|^3} + \frac12
    \frac{3(a\cdot b)^2 - a^2b^2}{|a|^5} + \mathcal O\pqty{ \frac{|b|^3}{|a|^4} },
  \end{align}
  valid for $a,b \in \R^3$ with $|a| \ge 4|b|$,
  and choose $a = (r-s)\till + x $ and $b = y  $.

  By our assumption that $\Omega^* = B(0,r_*)$, the monopole, the dipole and the quadrupole moments of
  $1_{\Omega^*} - \vartheta 1_{Q_\till}$ vanish.

  That is, for all $a \in \R^3 \backslash \Bqty{0}$, we have the equation
  \begin{align} \label{eq_multipole momenta vanish}
    \int_{\R^3} \pqty{ 1_{\tilo}(y) - \vartheta 1_{Q_\till}(y) }
    \pqty{ \frac1{|a|} + \frac{a\cdot y}{|a|^3} + \frac12
    \frac{3(a\cdot y)^2 - a^2y^2}{|a|^5} } \dd y = 0.
  \end{align}
  This follows from our neutrality condition
  \eqref{eq_local neutrality of Omega theta L} and the symmetries of a ball and a
  cube centered at $0$.
  More precisely, the function $1_{\tilo}(y) - \vartheta 1_{Q_\till}(y) $ is
  invariant under the reflection of one coordinate $y_i \mapsto -y_i$ as well as
  under the exchange of two coordinates $y_i$ and $y_j$.

  These symmetries cause the dipole, respectively the
  quadrupole moment to vanish. We stress that this is one of only two places
  where the additional assumption $\Omega^* = B(0,r_*)$ enters in our proof.
  The other one is the fact that $\Omega^* = B(0,r_*) \subset Q_l$ for
  $\vartheta \le \frac12$ if $\vartheta l^3 = |\Omega^*|$.
  We again refer to the Appendix for the necessary modifications to obtain an
  equation similar to \eqref{eq_multipole momenta vanish} in the absence of
  the assumption $\Omega^* = B(0,r_*)$.

  By \eqref{eq_multipole momenta vanish}, if we plug in the expansion \eqref{eq_taylor expansion multipole} and
  set $a = (r-s)\till + x$ and $b=y$, equation
  \eqref{eq_upper_bound_far_field1} is bounded from above by
  \begin{align} \nonumber
    &C \sum_{(r,s) \in V_{\textnormal{far}}}
    \int_{Q_\till} \int_{Q_\till} |1_{\tilo}(x) - \vartheta|
    \frac{|y|^3}{|r\till - s\till + x|^4} |1_{\tilo}(y) -\vartheta|
    \dd{x}\dd{y}
    \\ \label{eq_upper_bound_far_field2}
    \le C &\sum_{(r,s) \in V_{\textnormal{far}}}
    \int_{Q_\till} \frac{1_{\tilo}(x) +\vartheta}{|(r- s)\till + x|^4} \dd{x}
    \int_{Q_\till} \pqty{1_{\tilo}(y) + \vartheta} |y|^3 \dd{y}.
  \end{align}
  For $x \in Q_\till$ we have $|x| \le \sqrt{3}\, \till/2$. Since $|r-s| > M =
  10$, it follows $|(r-s)\tilde l + x| \ge \tilde l |r-s| - |x| \ge \frac12
  \tilde l |r-s|$. Equation \eqref{eq_upper_bound_far_field2} can thus be estimated from above by
  \begin{align} \nonumber
    C &\sum_{(r,s) \in V_{\textnormal{far}}} \frac{2^4}{\till^4|r - s|^4}
    \int_{Q_\till} \pqty{1_{\tilo}(x) + \vartheta} \dd{x}
    \pqty{ \int_{\tilo} |y|^3 \dd{y} + \vartheta \int_{Q_\till} |y|^3 \dd{y} }
    \\ \label{eq_upper_bound_far_field3}
    \le \frac{C}{\till^4} &\sum_{(r,s) \in V_{\textnormal{far}}} \frac1{|r - s|^4}
    \pqty{ 1 + \vartheta \till^{\;\!3}} \pqty{ 1 + \vartheta \till^{\;\!6} }  \le  C l^{-1}
    (1 + \vartheta l^3) (l^{-3} + \vartheta l^3) \sum_{(r,s) \in V_{\textnormal{far}}} \frac1{|r - s|^4},
  \end{align}
  where we again used $1 \le \lambda_\tl \leq 2$, and hence, $l \le \till \leq 2 l$,
  from \eqref{eq_bound on lambda tl crude}.

  Since $\vartheta l^3 = A^*$, it remains to evaluate the last sum over the set
  $V_{\textnormal{far}}$.  Recalling the bound on the number of boxes $N_\tl \leq
  \frac{L^3}{l^3}$, we have
  \begin{align}
    \sum_{(r,s) \in V_{\textnormal{far}}} \frac1{|r - s|^4}
    \le \sum_{r \in \mathcal C_\tl} \sum_{\substack{s\in
    \Z^3 \\ s\ne r}} \frac1{|r - s|^4}
    \leq \frac{L^3}{l^3} \sum_{\substack{s\in
    \Z^3\backslash \Bqty{0}}} \frac1{|s|^4}
    = C \vartheta L^3. \label{eq_upper_bound_far_field4}
  \end{align}
  Putting together
  \eqref{eq_upper_bound_far_field1}, \eqref{eq_upper_bound_far_field3} and \eqref{eq_upper_bound_far_field4} and
  using $\vartheta l^3 = A^*$, we obtain
  \begin{align*}
    &\frac1{\vartheta L^3}  \E_\tl^{\text{(far)}} \leq Cl^{-1} =
    C\vartheta^{1/3}.
  \end{align*}

  %%%%%%%%%%%%%%%%%%%%%%%%%%%%%%%%%%%%%%%%%%%%%%%%%%%%%%%%%%%%%%%%%%%

  \textit{Step 4: Conclusion.}    Inserting the bounds proved in Steps 1-3 back into \eqref{eq_splitting into self, near, far}, we obtain
  \[  \frac{\E_\tl[\Omega_\tl]}{\vartheta L^3} = \frac{\E^\text{(self)}_\tl + \E^\text{(near)}_\tl + \E^\text{(far)}_\tl }{\vartheta L^3} \leq \frac{E(A^*)}{A^*} + C \vartheta^{1/3} + C \frac{1}{\vartheta^{1/3} L}. \]
  The proof of Proposition \ref{pr_upper bound} is therefore complete.
\end{proof}

%%%%%%%%%%%%%%%%%%%%%%%%%%%%%%%%%%%%%%%%%%%%%%%%%%%%%%%%%%%%%%%%%%%

\section{Lower Bound of the Ground State Energy}

%%%%%%%%%%%%%%%%%%%%%%%%%%%%%%%%%%%%%%%%%%%%%%%%%%%%%%%%%%%%%%%%%%%

In this section, we give the proof of the lower bound from Theorem \ref{th_gse asymptotics}. Again, we restate here for convenience the result we want to prove.

\begin{proposition}[Lower Bound] \label{pr_lower bound}
  There is a constant $C > 0$ such that for all $\vartheta \in (0,1]$, $L >0$,
  \begin{align} \label{lower_bound}
    \frac{E_\tl}{\vartheta L^3} \,\geq\, \frac{E(A^*)}{A^*}  - C \vartheta^{1/5}.
  \end{align}
\end{proposition}

The proof of Proposition \ref{pr_lower bound} is based on reducing the problem
to a smaller length scale $1 \ll R \ll L$.

We define the Yukawa potential
\[ Y_{\omega}(x) = \frac{ \e^{-\omega|x|} }{|x|} \qquad \text{for} \quad x \in \R^3 \quad \text{and } \quad \omega > 0. \]
Using $Y_\omega$, we can bound the interaction part of $\mathcal E_\tl[\Omega]$ from below as follows.

\begin{lemma}[Lower bound on the interaction term]
  \label{le_localization_coulomb}
  There is $C > 0$ such that for all $L>0$, all $\vartheta \in [0,1]$, all $\omega > 0$ and all $\Omega\subset Q_L$, we have
  \begin{alignat*}{1}
    &\int_{\R^3} \int_{\R^3} \pqty{ 1_\Omega(x) - \vartheta 1_{Q_L}(x) }
    \frac1{|x-y|}\pqty{ 1_\Omega(y) - \vartheta 1_{Q_L}(y) } \dd{x} \dd{y}
    \\
    \ge & \int_{\R^3} 1_\Omega(x) Y_{\omega} (x-y)
    \, 1_\Omega(y) \dd x \dd y - C |\Omega| \vartheta \omega^{-2}.
  \end{alignat*}
\end{lemma}

\begin{proof}
First, we can estimate
  \begin{align*}
  & \qquad  \int_{\R^3} \int_{\R^3} \pqty{ 1_\Omega(x) - \vartheta 1_{Q_L}(x) }
    \frac1{|x-y|}\pqty{ 1_\Omega(y) - \vartheta 1_{Q_L}(y) } \dd{x} \dd{y} \\
   & \geq \int_{\R^3} \int_{\R^3} \pqty{ 1_\Omega(x) - \vartheta 1_{Q_L}(x) }
    Y_\omega(x-y) \pqty{ 1_\Omega(y) - \vartheta 1_{Q_L}(y) } \dd{x} \dd{y}, 
  \end{align*}
   because $\F[\frac1{|x|}] = \sqrt{\frac{\pi}{2}} \frac{1}{|k|^2} \ge
  \sqrt{\frac{\pi}{2}} \frac{1}{|k|^2 + \omega^2} = \F[Y_\omega](k)$, where $\F$ denotes the Fourier transform. Next, 
  \begin{align} \nonumber
    &\int_{\R^3} \! \int_{\R^3} \pqty{ 1_\Omega(x) - \vartheta 1_{Q_L}(x) }
    \,  Y_{\omega} (x-y)
    \, \pqty{ 1_\Omega(y) - \vartheta 1_{Q_L}(y) } \dd{x} \dd{y}
    \\
    \nonumber
    &\ge
    \int_{\R^3} \! \int_{\R^3} 1_\Omega(x) \,
    Y_{\omega} (x-y)  \, 1_\Omega(y) \dd x \dd y
    - 2 \vartheta \int_{\R^3} \! \int_{\R^3} 1_\Omega(x) Y_{\omega}(x-y) 1_{Q_L}(y)
    \dd x \dd y  \\
    &\ge
    \int_{\R^3} \! \int_{\R^3} 1_\Omega(x) Y_{\omega} (x-y)
    \, 1_\Omega(y) \dd x \dd y - C |\Omega| \vartheta \omega^{-2},
    \label{eq_indicator function smaller cubes}
  \end{align}
  where we bounded
  $\int_{Q_L} Y_{\omega}
  (x-y) \dd{y} \le \int_{\R^3} \frac{e^{-\omega |y|}}{|y|} \dd y \leq
  C\omega^{-2}$.
\end{proof}

We also need to control the behavior of the perimeter term under localization of $\Omega \subset Q_L$ to smaller boxes. The following lemma is useful for this purpose.

\begin{lemma}[Localization of the perimeter term] \label{le_localization_perimeter}
  Let $\Omega \subset \R^3$ have finite perimeter. Then for every $R > 0$,
  \[ \Per(\Omega) \geq \sum_{m \in \Z^3} \int_{Q_1} \Per(\Omega \cap Q_R(m + \mu)) \dd \mu  - \frac{6 |\Omega|}{R}. \]
\end{lemma}

\begin{proof}
  In every box, the boundary of $\Omega \cap Q_R(m + \mu)$ consists of two parts: the portion of $\partial \Omega$ lying inside $Q_R(m + \mu)$, and the portion of $\Omega$ intersecting $\partial Q_R(m + \mu)$, which is added by partitioning $\Omega$ into boxes. We therefore have that
  \begin{align} & \sum_{m \in \Z^3} \int_{Q_1}  \Per (\Omega \cap Q_R(m + \mu))
    \dd \mu \nonumber \\
    & \leq \int_{Q_1} \sum_{m \in \Z^3} \h^2(\partial \Omega \cap Q_R(m + \mu)) \dd
    \mu + \int_{Q_1} \sum_{m \in \Z^3} \h^2(\Omega \cap \partial Q_R(m + \mu)) \dd
    \mu \nonumber \\
    & \leq \Per (\Omega) + \int_{Q_1} \sum_{m \in \Z^3} \h^2(\Omega \cap \partial Q_R(m + \mu)) \dd \mu. \label{eq_localization perimeter}
  \end{align}
  Here, $\h^2$ denotes two-dimensional Hausdorff measure. It remains to evaluate the second term in \eqref{eq_localization perimeter}. Since all sets appearing there are subsets of faces of cubes, we can decompose
  \[ \bigcup_{m \in \Z^3} \Omega \cap \partial Q_R(m + \mu) = \bigcup_{i = 1}^3
  \bigcup_{l \in \Z} \Omega \cap \Bqty{ x \in \R^3\!:\, x_i = R\qty(l + \frac12 + \mu_i)}, \]
  i.e. we distinguish the 'slices' of $\Omega \cap \partial Q_R(m + \mu)$
  according to the coordinate hyperplane they are parallel to. Note that $\mathcal H^2$-almost every
  point in one hyperplane is contained in the boundary of exactly two cubes adjacent to
  the plane. Thus the union
  $\bigcup_{i=1}^3$ is disjoint up to an $\h^2$-null set and we obtain
  \[ \int_{Q_1} \sum_{m \in \Z^3} \h^2(\Omega \cap \partial Q_R(m + \mu)) \dd \mu
    = 2\sum_{i = 1}^3 \int _{\substack{[-1/2, 1/2]^3}} \dd \mu_1 \dd \mu_2 \dd \mu_3 \sum_{l \in \Z}
  \h^2(\Omega \cap \{ x_i = R(l + 1/2 + \mu_i) \} ). \]
The integrand on the right hand side only depends on \emph{one} of the $\mu_i$. We can therefore do the $\dd \mu_j$-integrations with $j \neq i$ to find that
  \begin{align*}
    & \int_{Q_1} \sum_{m \in \Z^3} \h^2(\Omega \cap \partial Q_R(m + \mu)) \dd \mu
    = 2\sum_{i = 1}^3 \int_{-1/2}^{1/2} \sum_{l \in \Z} \h^2(\Omega \cap \{ x_i =
    R(l + 1/2 + \mu_i) \} ) \dd \mu_i \\
  & = 2\sum_{i = 1}^3 \int_\R \h^2(\Omega \cap \{ x_i = R \mu_i \} ) \dd \mu_i = \frac{2}{R} \sum_{i = 1}^3\int_\R \h^2(\Omega \cap \{ x_i =  \mu_i \} ) \dd \mu_i = \frac{6 |\Omega|}{R} \end{align*}
    by Fubini's theorem. Plugging this in \eqref{eq_localization perimeter} completes the proof of Lemma \ref{le_localization_perimeter}.
\end{proof}

%%%%%%%%%%%%%%%%%%%%%%%%%%%%%%%%%%%%%%%%%%%%%%%%%%%%%%%%%%%%%%%%%%%

%%%%%%%%%%%%%%%%%%%%%%%%%%%%%%%%%%%%%%%%%%%%%%%%%%%%%%%%%%%%%%%%%%%

In the next lemma we combine the above estimates to obtain the crucial lower bound on the energy in terms of the auxiliary parameters $R$ and $\omega$. 

\begin{lemma} \label{le_bound on localized energy}
  For every $\Omega \subset Q_{L}$ with $|\Omega| > 0$ and every $R > 0$, we have that
  \[ \frac{ \E_\tl[\Omega]}{|\Omega| } \geq e^{-\sqrt{3} \omega R}
  \frac{E(A^*)}{A^*} - C \vartheta \omega^{-2} - \frac{6}{R}. \]
\end{lemma}

\begin{proof}
  Let $\Omega \subset Q_{L}$ and $R>0$.
  We split $\Omega$ into the (finite) disjoint union
  \begin{equation}
    \label{eq_splitting Omega 2nd localization}
    \Omega = \bigcup_{m \in \Z^3} \big( \Omega \cap Q_{R} (m+\mu_0) \big) =: \bigcup_{m \in \Z^3} \Omega^{(m)}
  \end{equation}
  for some $\mu_0 \in Q_1$ to be chosen below. Note that our choice of
  $\Omega^{(m)}$ in \eqref{eq_splitting Omega 2nd localization} ensures that
  $\text{diam}(\Omega^{(m)}) \leq \sqrt{3} R$.

  Hence, starting from Lemma \ref{le_localization_coulomb} and dropping the interactions between different boxes, we can
  estimate the energy from below as follows.
  \begin{align}
    \nonumber
    \E_\tl [\Omega] &\geq \sum_{m \in \Z^3} \qty( \Per (\Omega^{(m)}) +
    e^{-\omega \sqrt{3} R} \frac12 \iint_{ \Omega^{(m)} \times \Omega^{(m)}}
    \frac{\dd x \dd y}{|x-y|})  +  \mathcal P_R  - C |\Omega| \vartheta \omega^{-2} \\
    & \ge   e^{-\sqrt{3} \omega R} \sum_{m \in \Z^3} \E[\Omega^{(m)}] + \mathcal
    P_R - C |\Omega| \vartheta \omega^{-2}
    \label{eq_estimate energy E t K 2}
  \end{align}
  with the perimeter error term $\mathcal P_R := \Per(\Omega) - \sum_{m \in \Z} \Per (\Omega^{(m)})$.

  For every $m \in \Z^3$ with $|\Omega^{(m)}| >0$, we have
  $\frac{\E[\Omega^{(m)}]}{|\Omega^{(m)}|} \geq \frac{E(A^*)}{A^*}$, and hence
  \begin{equation}
    \label{eq_energy trick 1}
    \sum_{m \in \Z^3} \E[\Omega^{(m)}] = \sum_{m \in \Z^3, |\Omega^{(m)}| >0} \frac{\E[\Omega^{(m)}]}{|\Omega^{(m)}|} |\Omega^{(m)}| \geq \frac{E(A^*)}{A^*} \sum_{m \in \Z^3} |\Omega^{(m)}| = |\Omega| \frac{E(A^*)}{A^*}.
  \end{equation}
  Together with \eqref{eq_energy trick 1}, the lower bound \eqref{eq_estimate energy E t K 2} implies
  \begin{equation}
    \label{eq_estimate loc energy gs + perimeter}
    \E_{\tl}[\Omega] \geq e^{-\sqrt{3} \omega R} \frac{E(A^*)}{A^*} |\Omega|   +
    \mathcal P_R - C |\Omega| \vartheta \omega^{-2}.
  \end{equation}

  To bound the perimeter error $\mathcal P_R$ appropriately, recall from Lemma \ref{le_localization_perimeter} that we have the averaged estimate
  \begin{equation}
    \label{eq_estimate perimeter averaged}
    \int_{Q_1} \qty( \sum_{m \in \Z^3}  \Per(\Omega \cap Q_{R}(m + \mu)) ) \dd \mu \leq \Per(\Omega) + \frac{6 |\Omega|}{R},
  \end{equation}
  and therefore there exists $\mu_0 \in Q_1$ depending on $\Omega$ such that
  \begin{equation}
    \sum_{m \in \Z^3}  \Per(\Omega \cap Q_{R}(m + \mu_0)) \leq \Per[\Omega]  + \frac{6 |\Omega|}{R}.
  \end{equation}
  With this choice of $\mu_0$, we arrive at the bound
  \begin{equation}
    \label{eq_bound perimeter} \mathcal P_R = \Per(\Omega) - \sum_{m \in \Z} \Per (\Omega^{(m)}) \geq - \frac{6 | \Omega |}{R}.
  \end{equation}
  Combining \eqref{eq_estimate loc energy gs + perimeter} and \eqref{eq_bound perimeter} and dividing by $|\Omega|$, the statement of Lemma \ref{le_bound on localized energy} follows.
\end{proof}

%%%%%%%%%%%%%%%%%%%%%%%%%%%%%%%%%%%%%%%%%%%%%%%%%%%%%%%%%%%%%%%%%%%

It only remains to minimize the errors of the lower bound to the ground state energy $E_\tl$.

\begin{proof}[Proof of Proposition \ref{pr_lower bound}]
  Recalling that $|\Omega| = \vartheta L^3$, by Lemma \ref{le_bound on localized energy} we have
  \begin{equation}
    \label{eq_estimate GSE 3}
    \frac{E_\tl}{\vartheta L^3} \geq e^{-\sqrt{3} \omega R}  \frac{E(A^*)}{A^*}
    - C \vartheta \omega^{-2} - \frac{6}{R}.
  \end{equation}
  Since $e^{-x} \geq 1 - x$, from \eqref{eq_estimate GSE 3} we obtain
  \begin{align*}
    \frac{E_\tl}{\vartheta L^3}
    & \geq \qty(1 -  \sqrt{3} \omega R) \frac{E(A^*)}{A^*}
    - C \vartheta \omega^{-2} - \frac{C}{R} \\
    & \geq  \frac{E(A^*)}{A^*} - C \omega R - \frac{C}{R}
    - C \vartheta \omega^{-2}.
  \end{align*}
  Optimizing first in $R$, we take $R = \omega^{-1/2}$. With that choice, we have
  the inequality
  \[ \frac{E_\tl}{\vartheta L^3} \geq \frac{E(A^*)}{A^*}
  - C \omega^{1/2} - C \vartheta \omega^{-2}. \]
  Optimizing in $\omega$ gives $\omega = \vartheta^{2/5}$, and hence, we get
  \[ \frac{E_\tl}{\vartheta L^3} \geq \frac{E(A^*)}{A^*} - C \vartheta^{1/5}. \]
  The proof of Proposition \ref{pr_lower bound} is now complete.
\end{proof}

%%%%%%%%%%%%%%%%%%%%%%%%%%%%%%%%%%%%%%%%%%%%%%%%%%%%%%%%%%%%%%%%%%%

\begin{appendices}

  \section{Appendix: Removing the symmetry assumptions on $\Omega^*$}
  \label{appendix_removing_symmetry}

  %%%%%%%%%%%%%%%%%%%%%%%%%%%%%%%%%%%%%%%%%%%%%%%%%%%%%%%%%%%%%%%%%%%

  We give here the necessary modifications to obtain the upper bound from Theorem \ref{th_gse asymptotics} if one does not make any symmetry assumption on the energy-per-volume minimizer $\Omega^*$.

  In case $\vartheta > 1/C$, inequality \eqref{upper_bound intro} is equivalent to the
  bound
  \begin{align*}
    \frac{E_\tl}{\vartheta L^3} \le C.
  \end{align*}
  To show this, we do not need to use minimizers in the construction of our
  test set $\Omega_\tl$. Hence, it suffices to consider balls of any fixed
  radius $r_* > 0$ arranged on a lattice just as it is done
  in the proof of Proposition \ref{pr_upper bound}.
  Henceforth we may therefore assume
  \begin{align*}
    \vartheta \le \frac1C.
  \end{align*}

  The proof strategy of Theorem \ref{th_gse asymptotics} in the absence of symmetry of $\Omega^*$ is identical to the one of the upper bound in Section \ref{section upper bound}. One constructs a competitor set made from energy-per-volume minimizers $\Omega^*$ arranged on a lattice. The difficulty one faces is that in proving the error bound on the far-field interaction term, one cannot invoke the symmetry of $\Omega^*$ to prove that the monopole, dipole and quadrupole moments vanish as in \eqref{eq_multipole momenta vanish}.

  We resolve this difficulty by fine-adjusting the parameters of our lattice. More precisely, we show that the analogue of \eqref{eq_multipole momenta vanish} can still be achieved by considering a \emph{suitably translated and rotated} copy of $\Omega^*$, arranged on a \emph{slightly distorted} lattice.

  \paragraph{Notation.} To deal with cuboids instead of cubes, it is necessary to introduce some appropriate notation. For $r \in \R^3$ and $\textrm{\textbf{\textit{l}}} \in\R^3$, we define
  \begin{equation}
    \label{eq_definition cuboid}
    Q_\bl(r) := \{ x \in \R^3 \, : \, |x_i - l_i r_i| < l_i/2  \textnormal{ for } i\in\Bqty{1,2,3}\},
  \end{equation}
  and $Q_\bl := Q_\bl(0)$.
  Once again, pay attention to the fact that by definition, $Q_\bl(r)$ is the cuboid of
  side lengths $\bl$ centered at the point with coordinates $l_i r_i$, not
  centered at $r$! This is because we intend to cover $Q_L$ with many
  copies of the cuboid of side lengths $\bl$. Then, the parameter $r\in\Z^3$ simply
  counts those cuboids in each direction.

  More generally, given $\bm\lambda = (\lambda_1,
  \lambda_2, \lambda_3) \in \R^3$ and $\Omega \subset \R^3$, we define the 'inhomogeneous
  dilation' by $\bm\lambda$ of the set $\Omega$ to be $\bm\lambda \Omega:= \{
  (\lambda_1 x_1, \lambda_2 x_2, \lambda_3 x_3) \, : \, x \in \Omega \}$. Observe that with these definitions, one has $\bm \lambda Q_L = Q_{L \bm\lambda}$. We shall use both notations according to convenience.

  Furthermore, for
  $\Omega \in \bm\lambda Q_L$, we set
  \[ \E_{\tl, \bm\lambda}[\Omega] = \Per(\Omega) + \frac12 \int_{\bm\lambda Q_L} \int_{\bm\lambda Q_L} (1_\Omega(x) - \vartheta) |x-y|^{-1} (1_\Omega(y) - \vartheta) \dd x \dd y, \]
  and define the corresponding ground state energy by
  \[E_{\tl, \bm\lambda} = \inf \{ \E_{\tl, \bm\lambda} [\Omega] \, : \, \Omega
  \subset \bm\lambda Q_L,  |\Omega| = \vartheta |\bm\lambda Q_L| \}. \]

  With this notation at hand, we can prove the following two key lemmas.

  \begin{lemma}[Vanishing Multipole Moments]
    \label{le_vanishing multipole moments}
    Let $\Omega \subset \R^3$ be a bounded set. Assume that two numbers $l_0>0$
    and $\vartheta \in (0,1]$ are given such that $|\Omega| = \vartheta l_0^3$.
    If $\eta_0 := \frac{l_0}{\mathrm{diam}(\Omega)}$ is larger than a universal constant, then there is an orthogonal matrix $U \in \R^{3 \times 3}$, a translation vector $y \in \R^3$ and a scaling vector $\bl = \bm\lambda l_0$ such that the set $\Omega_0 := U (\Omega + y)$ is contained in $Q_\bl$ and satisfies
    \begin{align*} 0 &= \int_{\R^3} (1_{\Omega_0}(x) - \vartheta 1_{Q_\bl}(x)) \dd x =
      \int_{\R^3} x_i (1_{\Omega_0}(x) - \vartheta 1_{Q_\bl}(x))\dd x \\
      & = \int_{\R^3} (3 x_i x_j - \delta_{ij} |x|^2)  (1_{\Omega_0}(x) - \vartheta 1_{Q_\bl}(x))  \dd x
    \end{align*}
    for all $i, j \in \Bqty{1,2,3}$. Furthermore, the scaling parameters $\bm\lambda =
    (\lambda_1,\lambda_2,\lambda_3)$ satisfy
    \begin{align*}
      \lambda_1 \lambda_2 \lambda_3 = 1 \qquad \text{and} \qquad |\lambda_i - 1| \le C\eta_0^{-2} \qquad (i = 1,2,3)
    \end{align*}
    for a universal constant $C > 0$.
  \end{lemma}

  \begin{remark}
    We point out that since the proof below does not use the special form of
    $1_\Omega$ as an indicator function, the statement of Lemma \ref{le_vanishing
    multipole moments} remains true if one replaces $1_\Omega$ by an arbitrary charge distribution $\rho \geq
    0$, $\rho \in L^1(\R^3)$, with compact support.
  \end{remark}

  \begin{proof}[Proof of Lemma \ref{le_vanishing multipole moments}]
    Let $\Omega \subset \R^3$ satisfy $|\Omega| = \vartheta l_0^3$. We first observe that since rotations and translations do not change the volume $|\Omega| = \vartheta l_0^3$, we always have
    \[ 0 = \int_{\R^3} (1_{U(\Omega + y)}(x) - \vartheta 1_{Q_\bl}(x)) \dd x, \]
    as long as the constraint $\lambda_1 \lambda_2 \lambda_3 = 1$ is satisfied, which implies $|Q_\bl| = l_0^3$.

    Next, we claim that up to replacing $\Omega$ by its translate $\Omega + y$ for a suitable vector $y \in \R^3$, we may achieve that
    \begin{equation}
      \label{eq_dipole moment vanishes}
      0 = \int_{\R^3} x_i (1_{\Omega }(x) - \vartheta 1_{Q_\bl})\dd x  \qquad \text{for all } i = 1,2,3.
    \end{equation}
    for every $\bl \in \R^3$. Indeed, the cube $Q_\bl$ is symmetric with respect to the coordinate planes and thus
    \[ 0 =  \vartheta \int_{\R^3} x_i  1_{Q_\bl}(x) \dd x. \]
    Moreover, for $y \in \R^3$ one has
    \[ \int_{\R^3} 1_{\Omega + y} (x) x_i \dd x = \int_{\R^3} 1_\Omega(x) (x_i + y_i) \dd x = \int_\Omega x_i \dd x + y_i |\Omega|, \]
    Hence it suffices to set $y_i = - \frac{1}{|\Omega|} \int_\Omega x_i \dd x$. We continue for simplicity to denote the translated version $\Omega +y$ which satisfies \eqref{eq_dipole moment vanishes} by $\Omega$. Note also that if $\Omega$ satisfies \eqref{eq_dipole moment vanishes}, then so does $U \Omega$, for any invertible matrix $U \in \R^{3 \times 3}$.

    It remains to ensure the quadrupole moment to vanish by introducing appropriate $U \in \R^{3 \times 3}$ and $\bl = \bm\lambda l_0 \in \R^3$. Since the quadrupole moment of $\Omega$,
    \[ P = (P_{ij})_{i,j = 1,2,3} \qquad \text{with} \qquad P_{ij} := \int_{\Omega} (3 x_i x_j - \delta_{ij} |x|^2) \dd x, \]
    is a traceless symmetric $3 \times 3$-matrix with real entries, there is an orthogonal matrix $U \in \R^{3 \times 3}$ and numbers $a, b \in \R$ such that
    \begin{equation}
      \label{eq_diagonal quadrupole moment}
      \left[ {\begin{array}{ccc}
            a & 0 & 0 \\
            0 & b & 0 \\
            0 & 0 & - a - b
      \end{array} } \right]
      = U P U^T = \int_{\Omega} (3 (Ux)_i (Ux)_j - \delta_{ij} |Ux|^2) \dd x = \int_{U \Omega} (3 x_i x_j - \delta_{ij} |x|^2 ) \dd x.
    \end{equation}
    That is, up to replacing $\Omega$ by its rotated version $U \Omega =: \Omega_0$, whose monopole and dipole moments still vanish by the remarks made above, we can assume that its quadrupole moment is diagonal.

    To make the quadrupole moment of $(1_{\Omega_0} - \vartheta 1_{Q_\bl})$ vanish, we thus need to find a cuboid $Q_\bl$ of volume $|Q_\bl| =  l_0^3$ which contains $\Omega_0$ and satisfies
    \begin{equation}
      \label{eq_quadrupole}
      \begin{aligned}
        \vartheta \int_{Q_\bl} (3x_1^2 - |x|^2) \dd x = a, \\
        \vartheta \int_{Q_\bl} (3x_2^2 - |x|^2) \dd x = b.
      \end{aligned}
    \end{equation}

    Setting $l_1 = \lambda_1 l_0$, $l_2 = \lambda_2 l_0$ and $l_3 = \lambda_3 l_0 = \frac{l_0}{\lambda_1
    \lambda_2}$ (by the volume constraint), then by rescaling and using the relation $|\Omega| = \vartheta l_0^3$, the system \eqref{eq_quadrupole} is equivalent to

    \begin{equation} \label{eq_inverse function theorem}
      \begin{aligned}
        2 \lambda_1^2 - \lambda_2^2 - \frac{1}{\lambda_1^2 \lambda_2^2}
        &= \frac{12}{|\Omega| l_0^2}  a,
        \\
        - \lambda_1^2 + 2 \lambda_2^2 - \frac{1}{\lambda_1^2 \lambda_2^2}
        &= \frac{12}{|\Omega| l_0^2}  b.
      \end{aligned}
    \end{equation}
    By adding these two equations, respectively subtracting them, we obtain
    the equations

    \begin{equation} \label{equations_for_lambda}
      \begin{aligned}
        \lambda_1^2 + \lambda_2^2 - \frac{2}{\lambda_1^2 \lambda_2^2}
        &= \frac{12(a+b)}{|\Omega| l_0^2} =: 2c_1,
        \\
        \lambda_1^2 - \lambda_2^2 &= \frac{4(a-b)}{|\Omega| l_0^2} =: 2c_2.
      \end{aligned}
    \end{equation}
    Inserting the second equation of \eqref{equations_for_lambda} into the first one and changing to the center of mass
    coordinate $X = (\lambda_1^2 + \lambda_2^2)/2$ we get the equation
    \begin{align}
      X - \frac1{X^2 - c_2^2} &= c_1,
    \end{align}
    which is equivalent to the cubic equation
    \begin{align}
      p(X) := X^3 - c_1 X^2 - c_2^2 X - 1 + c_1c_2^2 &= 0.
    \end{align}
    It can be seen from \eqref{eq_diagonal quadrupole moment} that $|a+b| \leq 8 \text{ diam}(\Omega_0)^2 | \Omega_0|$. Therefore the
    definition \eqref{equations_for_lambda} of the $c_i$ implies that $|c_i| <
    48 \text{ diam}(\Omega_0)^2 l_0^{-2}$. Hence, if $\eta_0 = \frac{l_0}{\text{diam}(\Omega_0)}$ is large enough, the polynomial $p$ will be
    very close to $X^3-1$. Since $X^3 - 1$ has exactly one complex zero close to $1$
    (namely $1$), we can apply Rouch\'e's theorem in a ball of radius $\sim \eta_0^{-2}$ around $1$. Hence there exists exactly one complex zero $X_0$ of $p$
    with $|X_0 - 1| \le C \eta_0^{-2}$. Since the coefficients of $p$ are real, uniqueness
    of the zero implies that $X_0$ is in fact real.

    We therefore get solutions $\lambda_1, \lambda_2 >0$ of
    \eqref{equations_for_lambda} which satisfy
    \begin{align*}
      |\lambda_1 - 1| = \frac{|\lambda_1^2 - 1|}{\lambda_1 + 1}
      = \frac{|X_0 + c_2 -1|}{\lambda_1 + 1} \le C \eta_0^{-2},
      \\
      |\lambda_2 - 1| = \frac{|\lambda_1^2 - 1|}{\lambda_1 + 1}
      = \frac{|X_0 - c_2 -1|}{\lambda_1 + 1} \le C \eta_0^{-2}.
    \end{align*}
    Note that $\lambda_3 = 1/(\lambda_1\lambda_2)$ also satisfies $|\lambda_3 - 1|
    \le C \eta_0^{-2}$. Moreover, the fact that $\int_{\Omega_0} x_i \dd x = 0$ implies easily that $\Omega_0 \subset Q_{l_0}$ for every $l_0 \geq 2 \, \text{diam}(\Omega)$. This completes the proof of Lemma \ref{le_vanishing multipole moments}.
  \end{proof}

  Our next lemma shows that for $\bm \lambda$ close to $(1,1,1)$, we can replace the ground state energy of $Q_L$ by that of the cuboid $Q_{\bm \lambda L}$ with only a small error.

  \begin{lemma}[Approximating $E_\tl$ by a cuboid $E_{\tl, \bm\lambda}$]
    \label{le_approximation by cuboid}
    Suppose that $\bm\lambda = (\lambda_1, \lambda_2, \lambda_3)$ is such that $\lambda_1 \lambda_2 \lambda_3 = 1$ and assume that $\lambda_i \in [1- \delta, 1+ \delta]$ for $i = 1,2,3$, for some $\delta \in [0,1]$. Then we have
    \begin{equation}
      \label{eq_claimed bound on cuboid GSE}
      \E_\tl[\Omega] \leq (1+ C \delta) \E_{\tl, \bm\lambda}[ \bm\lambda \Omega], \qquad \text{for all } \Omega \subset Q_L,
    \end{equation}
    where $C >0$ is a universal constant independent of $\delta$, $\vartheta$, $L$ and $\Omega$.
    In particular, this implies
    \begin{equation}
      \label{eq_claimed bound on GSE after infimum}
      E_\tl \leq (1+ C \delta) E_{\tl, \bm\lambda}.
    \end{equation}
  \end{lemma}

  \begin{proof}
    Let $\Omega \subset Q_L$ arbitrary and consider, for $\bm\lambda$ as in the statement, the set $\bm\lambda \Omega$. Note that since $\lambda_1 \lambda_2 \lambda_3 = 1$, we have $|\bm\lambda \Omega| = |\Omega|$ and $|\bm\lambda Q_{L}| = |Q_L|$.

    To prove \eqref{eq_claimed bound on cuboid GSE}, we consider the perimeter and Coulomb terms separately. Let us assume for definiteness in the following that $\lambda_1 \leq \lambda_2 \leq \lambda_3$. Firstly, recall the definition
    \begin{equation}
      \label{eq_definition perimeter measure theoretic}
      \Per(\Omega) = \sup \Bqty{ \int_\Omega \operatorname{div} \varphi(x) \dd x \, : \, \varphi \in C^1_c ( \R^3, \R^3), \| \varphi\|_\infty \leq 1 }.
    \end{equation}
    For any $\varphi$ as in \eqref{eq_definition perimeter measure theoretic} and $\bm\lambda \in \R^3$ with $\lambda_1 \lambda_2 \lambda_3 = 1$, define the vector field $\varphi_{\bm\lambda} \in C^1_c ( \R^3, \R^3)$ by setting its $i$-th component to be $\varphi_{\bm\lambda, i}(x) = \lambda_i \varphi_i(\lambda_1^{-1} x_1, \lambda_2^{-1} x_2, \lambda_3^{-1} x_3)$.
    One easily checks that
    \begin{equation}
      \label{eq_divergence identity}
      \int_\Omega \operatorname{div} \varphi(x) \dd x = \int_{\bm\lambda \Omega}
      \operatorname{div} \varphi_{\bm\lambda} (x) \dd x =
      \|\varphi_{\bm\lambda}\|_\infty  \int_{\bm\lambda \Omega} \operatorname{div}
      \frac{\varphi_{\bm\lambda}(x)}{\|\varphi_{\bm\lambda}\|_\infty }  \dd x.
    \end{equation}
    Moreover, we estimate
    \begin{equation}
      \label{eq_ell infty estimate}
      \|\varphi_{\bm\lambda}\|^2_\infty = \sup_{x \in \R^3} \sum_{i=1}^3 \lambda_i^2
      \varphi_i^2(\lambda_1^{-1} x_1, \lambda_2^{-1} x_2, \lambda_3^{-1} x_3) \leq
      \lambda_3^2 \|\varphi\|^2_\infty \leq 1 + C\delta.
    \end{equation}
    In view of the definition \eqref{eq_definition perimeter measure theoretic} of the perimeter, we can take the $\sup$ over all $\varphi \in C^1_c(\R^3, \R^3)$ with $\|\varphi\|_\infty \leq 1$ to obtain
    \begin{equation}
      \label{eq_cuboid estimate perimeter total}
      \Per(\Omega) = \sup_\varphi \int_\Omega \operatorname{div} \varphi(x) \dd x =
      \sup_\varphi \|\varphi_{\bm\lambda}\|_\infty  \int_{\bm\lambda \Omega}
      \operatorname{div}
      \frac{\varphi_{\bm\lambda}(x)}{\|\varphi_{\bm\lambda}\|_\infty }  \dd x \leq (1
      + C\delta) \Per(\bm\lambda \Omega),
    \end{equation}
    where we have used \eqref{eq_ell infty estimate} for the last inequality.

    To estimate the Coulomb term, it is convenient to pass to the Fourier representation. Set $f(x) := 1_\Omega(x) - \vartheta 1_{Q_L}(x)$, then
    \[ 1_{\bm\lambda \Omega} (x) - \vartheta 1_{\bm\lambda Q_L}(x) =
    f(\lambda_1^{-1} x_1, \lambda_2^{-1} x_2, \lambda_3^{-1} x_3) =: f_{\bm\lambda}(x), \]
    and one easily computes that $\F{f_{\bm\lambda}}(p) =
    \F{f}(\lambda_1 p_1, \lambda_2 p_2, \lambda_3 p_3)$. Therefore we have
    \begin{equation}
      \label{eq_cuboid estimate coulomb}
      \begin{aligned}
        & \frac{1}{4 \pi} \int_{Q_L} \int_{Q_L} \frac{(1_{\Omega}(x)-\vartheta)(1_{ \Omega}(y)-\vartheta)}{|x-y|}\dd x \dd y = \int_{\R^3} \frac{|\F{f}(p)|^2}{p^2} \dd p = \int_{\R^3} \frac{|\F{f_{\bm\lambda}}(p)|^2}{\sum_{i=1}^3 \lambda_i^{2} p_i^2} \dd p \\
        & \leq \lambda_1^{-2} \int_{\R^3} \frac{|\F{f_{\bm\lambda}}(p)|^2}{p^2} \dd p \leq
        (1 + C \delta) \frac{1}{4 \pi} \int_{\bm\lambda Q_L} \int_{\bm\lambda Q_L} \frac{(1_{\bm\lambda \Omega}(x)-\vartheta)(1_{ \bm\lambda \Omega}(y)-\vartheta)}{|x-y|}\dd x \dd y.
      \end{aligned}
    \end{equation}

    Combining estimates \eqref{eq_cuboid estimate perimeter total} and \eqref{eq_cuboid estimate coulomb}, the proof of \eqref{eq_claimed bound on cuboid GSE} is complete.

    The bound \eqref{eq_claimed bound on GSE after infimum} on the ground state energy follows from \eqref{eq_claimed bound on cuboid GSE} simply by taking the infimum over all $\Omega \subset Q_L$ with $|\Omega| = \vartheta L^3$. The proof of Lemma \ref{le_approximation by cuboid} is therefore complete.
  \end{proof}

  %%%%%%%%%%%%%%%%%%%%%%%%%%%%%%%%%%%%%%%%%%%%%%%%%%%%%%%%%%%%%%%%%%%

  Using Lemmas \ref{le_vanishing multipole moments} and \ref{le_approximation by cuboid}, we are now ready to give the proof of the upper bound from Proposition \ref{pr_upper bound} without assuming any symmetry on $\Omega^*$. Since most parts are identical to the proof in Section \ref{section upper bound}, we only give the necessary modifications in the construction of the competitor set at the beginning of the proof.

  \begin{proof}[Proof of Proposition \ref{pr_upper bound} without symmetry of $\Omega^*$]
    As in the proof given in Section \ref{section upper bound}, let us set $l_0 : = {A^*}^{1/3} \vartheta^{-1/3}$ to be the characteristic length of the small boxes. Let $\Omega^*$ be some set satisfying $|\Omega^*| = A^*$ and $\mathcal E [ \Omega^*] = E(A^*)$. We may assume (up to changing $\Omega^*$ on a null-set) that $\text{diam}(\Omega^*) < \infty$, see  \cite[Lemma 4.1]{KnMu} and \cite[Lemma 4]{LuOt}. By Lemma \ref{le_vanishing multipole moments}, there are $U \in \R^{3 \times 3}$ orthogonal, $y \in \R^3$ and $\bm\lambda \in \R^3$ with $|\lambda_i - 1| \le Cl_0^{-2}$ and $\lambda_1 \lambda_2 \lambda_3 = 1$ such that setting $\bl = \bm\lambda l_0$, the set $\Omega^*_0 := U(\Omega^* + y)$ is contained in $Q_\bl$ and satisfies
    \begin{equation}
      \label{eq_multipole moments vanish proof of thm}
      0 = \int_{\R^3} (1_{\Omega^*_0}(x) - \vartheta 1_{Q_\bl}) \dd x =
      \int_{\R^3} x_i (1_{\Omega^*_0}(x) - \vartheta 1_{Q_\bl})\dd x = \int_{\R^3} (3 x_i x_j - \delta_{ij} |x|^2)  (1_{\Omega^*_0}(x) - \vartheta 1_{Q_\bl})  \dd x.
    \end{equation}
    By Lemma \ref{le_approximation by cuboid}, we have
    \begin{equation}
      \label{eq_energy error cuboid} E_\tl \leq (1+ C l_0^{-2}) E_{\tl, \bm\lambda} = (1 + C \vartheta^{2/3}) E_{\tl, \bm\lambda} .
    \end{equation}
    To prove the upper bound from Theorem \ref{th_gse asymptotics}, it therefore suffices to prove the upper bound
    \begin{equation}
      \label{eq_upper bound cuboid proof of thm} \frac{E_{\tl, \bm\lambda}}{\vartheta L^3} \leq  \frac{E(A^*)}{A^*}
      + C \vartheta^{1/3} + \frac{C}{\vartheta^{1/3} L}
    \end{equation}
    because the additional error term coming from the estimate \eqref{eq_energy error cuboid} is subleading.

    To prove \eqref{eq_upper bound cuboid proof of thm}, we construct a competitor
    set by placing copies of the set $\Omega_0^*$ in boxes $Q_\bl(r)$, $r \in \Z^3$.
    Let $\mathcal C_\bl = \{ r \in \Z^3 \, : \, Q_\bl(r) \subset \bm \lambda Q_L
    \}$ be the set of lattice points $r$ such that the cubes $Q_\bl(r)$ are fully contained in $\bm \lambda Q_L$. Then, setting
    \begin{equation}
      \lambda_{\tl, \bm \lambda}^3 = \frac{\vartheta L^3}{A^* |\mathcal C_\bl|},
    \end{equation}
    we obtain $\bm\lambda Q_L$ as a union of the boxes $Q_{\lambda_{\tl, \bm
    \lambda} l_0}(r)$. That is, we can cover the large box exactly by an integer number
    of small boxes. We therefore define
    \[ \Omega_{\tl, \bm\lambda} = \bigcup_{r \in \mathcal C_\bl} (\lambda_{\tl, \bm
      \lambda} \Omega_0^* + \bl r),
      \qquad\qquad\textnormal{where}\quad
    \bl r := (l_1r_1, l_2r_2, l_3r_3). \]
    Note that this definition fulfills the mass constraint
    \begin{equation} \label{eq_mass constraint with lambda appendix}
      |\Omega_\tl| = |\mathcal C_\bl| A^* \lambda_\tl^3 = \vartheta L^3.
    \end{equation}
    The proof of Theorem \ref{th_gse asymptotics} can now be finalized by following exactly the same steps as in the proof of the upper bound of Theorem \ref{th_gse asymptotics}, using the vanishing of the multipole moments from \eqref{eq_multipole moments vanish proof of thm} in the bound on the far-field interaction. We omit the remaining details.
  \end{proof}

\section{Appendix: Independence of the thermodynamic limit from boundary conditions}
\label{se_thermodynamic_limit}

In this appendix, we prove that the value and the existence of the thermodynamic
limit is independent of the choice of boundary conditions. Although we are confident that in our setting the existence of the thermodynamic limit could be proved with the classical methods from \cite{LeLi, LiNa}, it seems hard to study different boundary conditions or obtain good asymptotics in this framework.
We therefore rely here strongly on a more recent result by Alberti, Choksi and Otto \cite{AlChOt}, who prove convergence of
the energy per volume with an essentially optimal convergence rate of $L^{-1}$ in the
case of Dirichlet and Neumann boundary conditions. We show that their
conclusion also holds for periodic boundary conditions, as well as for the
'whole space' boundary conditions which we use in this paper. In fact, we will see that the Dirichlet and Neumann boundary conditions give the lowest, respectively highest, interaction energy to any set $\Omega$ among a large class of boundary conditions, see Lemmas \ref{lemma dirichlet leq square} and \ref{lemma square leq N} and Remark \ref{remark other boundary conditions} below.

Take some $\vartheta \in (0,1)$ which will be fixed throughout the following and not reflected in the notation.

We consider the functional
\begin{equation}
  \label{definition liquid drop functional}
  \E_{\#,L}[\Omega] := \Per(\Omega) + \frac{1}{2} \int_{Q_L} \int_{Q_L} (1_\Omega(x) - \vartheta) G_{\#,L}(x,y) (1_\Omega(y)- \vartheta) \dd x \dd y
\end{equation}
and the corresponding ground state energy on the cube $Q_L = [-L/2, L/2]^3$
\begin{equation}
\label{definition ground state functional}
E_{\#,L} := \inf\{ \E_\#[\Omega] \,:\, \Omega \subset Q_L, \, |\Omega| = \vartheta L^3 \}.
\end{equation}
Here, the symbol $\#$ is a placeholder for the indices $D$, $N$, $P$ and $\infty$
denoting Dirichlet, Neumann, periodic or 'whole space' boundary conditions, respectively.
% For $\# \in \{D, N, P, \infty \}$,
We define $G_{\#,L}$ to be the Green's function for the problem on $Q_L$ with the
respective boundary condition.

Note that our energy functional $\E_\tl$ defined in \eqref{eq_definition
functional E theta L} is a rescaled version of the whole space functional
$\E_{\infty, L}$. Indeed, the Green's function of the whole space problem is
$G_{\infty, L}(x,y) = \frac1{4\pi} \frac{1}{|x-y|}$ and thus, we have
for any $\Omega \subset Q_L$ the equality
\begin{align*}
  \E_\tl[\Omega] = (4\pi)^{-2/3} \E_{\infty,\, (4\pi)^{1/3}L}\bqty{(4\pi)^{1/3}
  \Omega}.
\end{align*}

Finally, we denote the thermodynamic limits of the respective boundary condition
by
\[ e_\# = \lim_{L \to \infty} \frac{E_{\#,L}}{L^3}. \]

The main result of this Appendix is that the thermodynamic limit $e_\#$ exists and is independent of the boundary conditions.

\begin{theorem} \label{theorem independence from boundary values}
There are constants $\sigma^* >0$ and $C>0$ depending only on $\vartheta$ such that
\[ \abs{ \frac{E_{\#,L}}{L^3} - \sigma^* } \leq \frac{C}{L}  \]
for every $L \geq C$ and every $\# = D, N, P, \infty$. In particular,
\[ e_D = e_N = e_P = e_\infty = \sigma^*. \]
\end{theorem}

\begin{proof}
It has been proved in \cite{AlChOt} that there exist $\sigma^* >0$ and $C >0$ such that
\begin{equation}
\label{ACO inequalities} \sigma^* - \frac{C}{L} \leq \frac{E_{D,L}}{L^3} \leq  \frac{E_{N,L}}{L^3} \leq \sigma^* + \frac{C}{L}
\end{equation}
for all $L \geq C$. We will complete the proof of Theorem \ref{theorem independence from boundary values} by showing that for every $L>0$ and every $\Omega \subset Q_L$ with $|\Omega| = \vartheta L^3$ the following inequalities between different boundary conditions hold.
\begin{enumerate}
\item[(a)] $\E_{D, L}[\Omega] \leq \E_{\#, L}[\Omega]$ for all $\# \in \{ N, P, \infty\}$.
\item[(b)] $\E_{\#, L}[\Omega] \leq \E_{N,L}[\Omega]$ for all $\# \in \{D, P, \infty\}$.
\end{enumerate}
We point out that the validity of these inequalities requires neither minimization nor passage to the thermodynamic limit. We will prove the inequalities (a) and (b) in Lemmas \ref{lemma dirichlet leq square} and \ref{lemma square leq N} below. Theorem \ref{theorem independence from boundary values} is then an immediate consequence of \eqref{ACO inequalities} and inequalities (a) and (b).
\end{proof}

The proofs of Lemmas \ref{lemma dirichlet leq square} and \ref{lemma square leq N} rely on rewriting the interaction energy in \eqref{definition liquid drop functional} as a gradient integral. This is inspired by the approach in \cite{AlChOt}. Given an admissible set $\Omega \subset Q_L$ with $|\Omega| = \vartheta L^3$, we define
\begin{equation}
\label{definition v of omega}
v_{\#,L}^{(\Omega)}(x) := \int_{Q_L} G_{\#,L}(x,y) (1_\Omega(y) - \vartheta) \dd y, \quad x \in Q_L, \qquad \text{ for  } \# \in \{D, N, P \}
\end{equation}
and
\[ v_{\infty, L}^{(\Omega)}(x) := \frac{1}{4 \pi} \int_{Q_L} \frac{1_\Omega(y) - \vartheta}{|x-y|}  \dd y, \quad x \in \R^3, \qquad \text{for } \# = \infty, \]
the \emph{potential} associated with $\Omega$. The function $v_{\#,L}^{(\Omega)}$ satisfies
\[ - \Delta v_{\#,L}^{(\Omega)}  = 1_\Omega - \vartheta \quad \text{ on } Q_L \]
together with the boundary conditions given by the choice of $\#$. Notice that we even have
\[ -\Delta v_\infty^{(\Omega)} = 1_\Omega - \vartheta 1_{Q_L} \quad \text{ on all of  } \R^3. \]

Using integration by parts together with the respective boundary conditions, we can now write the energy $\E_{\#,L}[\Omega]$ of any set $\Omega \subset Q_L$ with $|\Omega| = \vartheta L^3$ as
\begin{equation}
\label{energy reformulation d, n, p}
\E_{\#,L}[\Omega] = \Per(\Omega) + \frac12 \int_{Q_L} |\nabla v_{\#,L}^{(\Omega)}(x)|^2 \dd x \qquad \text{for  } \# = D, N, P,
\end{equation}
respectively
\begin{equation}
\label{energy reformulation infty}
\E_{\infty, L}[\Omega] = \Per(\Omega) + \frac12 \int_{\R^3} |\nabla v_{\infty, L}^{(\Omega)}(x)|^2 \dd x \qquad \text{for  } \# = \infty.
\end{equation}

We are now ready to prove the inequalities (a) and (b) appearing in the proof of Theorem \ref{theorem independence from boundary values}.

\begin{lemma} \label{lemma dirichlet leq square}
Let $L > 0$ and $\Omega \subset Q_L$ with $|\Omega| = \vartheta L^3$. Then $\E_{D,L}[\Omega] \leq \E_{\#,L}[\Omega]$ for all $\# \in \{N, P, \infty\}$.
\end{lemma}

\begin{proof}
Let $\Omega \subset Q_L$ with $|\Omega| = \vartheta L^3$, let $\# \in \{ N, P, \infty \}$ and abbreviate $v_D = v_{D,L}^{(\Omega)}$ and $v_\# = v_{\#, L}^{(\Omega)}$. Then we have
\begin{align*}
\int_{Q_L} |\nabla v_\#|^2 &= \int_{Q_L} |\nabla v_D|^2 + \int_{Q_L} |\nabla (v_D - v_\#)|^2 - 2 \int_{Q_L} (\nabla v_D - \nabla v_\#) \cdot \nabla v_D \\
&= \int_{Q_L} |\nabla v_D|^2 + \int_{Q_L} |\nabla (v_D - v_\#)|^2 - 2 \int_{\partial Q_L} \frac{\partial (v_D-v_\#)}{\partial \nu} v_D \geq \int_{Q_L} |\nabla v_D|^2,
\end{align*}
where we used $\Delta v_D = \Delta v_\#$ on $Q_L$ and the fact that $v_D$ vanishes on $\partial Q_L$. We conclude that
\[ \E_{D,L}[\Omega] = \Per(\Omega) + \int_{Q_L} |\nabla v_D|^2 \leq \Per(\Omega) + \int_{Q_L} |\nabla v_\#|^2 \leq \E_{\#,L}[\Omega] \]
since the perimeter term does not depend on the boundary condition.
\end{proof}

\begin{lemma} \label{lemma square leq N}
Let $L > 0$ and $\Omega \subset Q_L$ with $|\Omega| = \vartheta L^3$. Then $\E_{\#,L}[\Omega] \leq \E_{N,L}[\Omega]$ for all $\# \in \{D, P, \infty\}$.
\end{lemma}

\begin{proof}
Let $\Omega \subset Q_L$ with $|\Omega| = \vartheta L^3$, let $\# \in \{D,  P, \infty \}$ and abbreviate $v_\# = v_{\#,L}^{(\Omega)}$ and $v_N = v_{N, L}^{(\Omega)}$.
Performing the same calculation as in the proof of Lemma \ref{lemma dirichlet leq square}, we obtain
\begin{align} \nonumber
\int_{Q_L} |\nabla v_N|^2 &= \int_{Q_L} |\nabla v_\#|^2 + \int_{Q_L}
|\nabla (v_\# - v_N)|^2 - 2 \int_{Q_L} (\nabla v_\# - \nabla v_N)
\cdot \nabla v_\# \\ \label{neumann_upper_bound_1}
&= \int_{Q_L} |\nabla v_\#|^2 + \int_{Q_L} |\nabla (v_\# - v_N)|^2 - 2 \int_{\partial Q_L} \frac{\partial v_\#}{\partial \nu} v_\#,
\end{align}
where we used $\Delta v_N = \Delta v_\#$ on $Q_L$ and the fact that $\frac{\partial v_N}{\partial \nu} = 0$ on $\partial Q_L$.
The Dirichlet and the periodic boundary condition imply that
\[ - \int_{\partial Q_L} \frac{\partial v_D}{\partial \nu} v_D = -
\int_{\partial Q_L} \frac{\partial v_P}{\partial \nu} v_P = 0. \]
Therefore, we deduce from equation \eqref{neumann_upper_bound_1} the bound
\[ \int_{Q_L} |\nabla v_N|^2  \geq \int_{Q_L} |\nabla v_\#|^2 \]
for $\# \in \{D, P\}$.
In case $\# = \infty$, integration by parts yields
\begin{align*} \label{neumann_upper_bound_2}
- 2\int_{\partial Q_L} \frac{\partial v_\infty}{\partial \nu} v_\infty
= 2\int_{\R^3 \setminus Q_L} |\nabla v_\infty|^2 \geq 0,
\end{align*}
since $\Delta v_\infty = 0$ on $\R^3 \setminus Q_L$.
Plugging this in equation \eqref{neumann_upper_bound_1}, we get
\begin{align*}
 \int_{Q_L} |\nabla v_N|^2  \geq \int_{\R^3} |\nabla v_\infty|^2.
\end{align*}
The lemma follows since the perimeter term does not depend on the boundary
condition.
\end{proof}

\begin{remark}
\label{remark other boundary conditions}
In fact, the proofs of Lemmas \ref{lemma dirichlet leq square} and \ref{lemma
square leq N} show that we can treat any self-adjoint boundary condition on $\partial Q_L$ with the property that $- \int_{\partial Q_L} \frac{\partial v}{\partial \nu} v \geq 0$ for
every $v$ satisfying this boundary condition. A large such class is for example given by the Robin boundary condition
\[ \frac{\partial v}{\partial \nu}(x)  = - \beta\qty(\frac{x}{L}) v(x) \qquad \text{ on } \partial Q_L, \]
where $\beta \in L^\infty(\partial Q_1)$ is an arbitrary nonnegative bounded function on $\partial Q_1$.
\end{remark}

  %%%%%%%%%%%%%%%%%%%%%%%%%%%%%%%%%%%%%%%%%%%%%%%%%%%%%%%%%%%%%%%%%%%

  \paragraph{Acknowledgements.} The authors acknowledge partial support by the
  U.S. National Science Foundation through grant DMS-1363432 (R.L.F.).

\end{appendices}

%%%%%%%%%%%%%%%%%%%%%%%%%%%%%%%%%%%%%%%%%%%%%%%%%%%%%%%%%%%%%%%%%%%
%%%%%%%%%%%%%%%%%%%%%%%%%%%%%%%%%%%%%%%%%%%%%%%%%%%%%%%%%%%%%%%%%%%

\bibliography{liquid_drop_model}
  \bibliographystyle{plain}

(Lukas Emmert) \textsc{Mathematisches Institut, Ludwig-Maximilians Universität München, Theresienstr. 39, 80333 München, Germany}

\textit{Email address:}
\href{mailto:lukas.emmert@lmu.de}{\texttt{lukas.emmert@lmu.de}}

~

(Rupert L. Frank) \textsc{Mathematisches Institut, Ludwig-Maximilians Universität München, Theresienstr. 39, 80333 München, Germany, and Department of Mathematics, California
Institute of Technology, Pasadena, CA 91125, USA}

\textit{Email address:} \href{mailto:r.frank@lmu.de}{\texttt{r.frank@lmu.de}}

~

(Tobias König) \textsc{Mathematisches Institut, Ludwig-Maximilians Universität München, Theresienstr. 39, 80333 München, Germany}

\textit{Email address:} \href{mailto:tkoenig@math.lmu.de}{\texttt{tkoenig@math.lmu.de}}

\end{document}